\documentclass{article}

\usepackage{amsfonts}
\usepackage{amssymb,epic,eepic}
\usepackage{amsmath}
\usepackage{dcolumn}
\usepackage{bm}
\usepackage{bbm}
\usepackage{verbatim}
\usepackage{color}
\usepackage{amsthm}
\usepackage{marvosym}
\usepackage{pdfpages}
\usepackage{sidecap}
\usepackage[font=footnotesize,width=12cm,labelfont=bf]{caption}
\newcommand{\hr}{{\mathcal H}}
\newcommand{\cn}{{\mathcal N }}
\newcommand{\E}{{\mathcal E }}

\newcommand{\kr}{{\mathcal K}}

\newcommand{\rr}{{\mathbb R}}

\newcommand{\nn}{{\mathbb N}}

\newcommand{\B}{\mathcal B}

\newcommand{\bX}{\mathbf X}
\newcommand{\eins}{{\mathbbm{1}}}
\newcommand{\spec}{\mathrm{spec}}

\newtheorem{theorem}{Theorem}
\newtheorem*{acknowledgement}{Acknowledgement}

\newtheorem{lemma}{Lemma}

\newtheorem{remark}{Remark}

\newcommand{\tr}{\mathrm{tr}}
\newcommand{\supp}{\mathrm{supp}}

\begin{document}
\title{The depolarising channel and Horns problem}
\author{Janis N\"otzel \\
\scriptsize{Electronic address: janis.noetzel@tum.de}
\vspace{0.2cm}\\
$^{1}$ {\footnotesize Theoretische Informationstechnik, Technische Universit\"at M\"unchen,}\\
{\footnotesize 80290 M\"unchen, Germany}
}
\maketitle
\begin{abstract}
We investigate the action of the depolarising (qubit) channel on permutation invariant input states. More specifically, we raise the question on which invariant subspaces the output of the depolarising channel, given such special input, is supported. An answer is given for equidistributed states on isotypical subspaces, also called symmetric Werner states. Horns problem and two of the corresponding inequalities are invoked as a method of proof.
\end{abstract}

\begin{section}{Introduction}
This is yet another brick we build in order to get a deeper understanding of the connections between quantum information theory and representation theory of the symmetric group. Despite considerable work (see \cite{calderbank-rains-shor-sloane} and references therein to get a feeling for error correction - error correcting codes naturally work as entanglement transmission codes for the depolarising channel, and \cite{ouyang}, \cite{fern-whaley} for more recent approaches) on the topic, the question when exactly the entanglement transmission capacity of the depolarising (qubit) channel is greater than zero is still unanswered. Yet, this channel is one of the simplest nontrivial models of a quantum channel one could think of. An answer to above question could also give hints to a solution of the very same question for more complex channel models (as the one considered in \cite{abbn}, for example).\\
This work provides bounds on the Hilbert Schmidt scalar product between normalized projections onto isotypical subspaces of the symmetric group and the output states of the depolarising (qubit) channel, given an input of the very same structure. Results are given in dependence of the depolarising parameter.\\
The method of proof is to decompose the action of the channel into convex combinations of channels that act as identity on a certain number of subsystems and as `useless channel' on the remaining ones.\\
It is, actually, the deeper one of our two results to calculate when exactly above mentioned Hilbert Schmidt scalar product is equal to zero, when the channel under consideration is given by some product of identity channel and `useless channel'. Here, we invoke the asymptotical connection between Littlewood-Richardson coefficients and Horns problem that was first noted in \cite{lidskii}, proven in \cite{klyachko} and brought to our attention through the work of \cite{christandl-cc}, who gave a proof resting solely on quantum information theoretic tools. In a second step we then use two of the inequalities emerging in Horns problem \cite{horn}.\\
We should note that our original intent was to prove this result for normalized projections onto \emph{irreducible} subspaces, not only \emph{isotypical} ones and that the weaker formulation we give here is only due to the fact that the right tools to deal with that question seem to be missing.\\
However, we find that the connection between the three different worlds that shows up in our proof is of independent interest. Our Theorem \ref{mainresult} can be interpreted as a 'reverse' variant of a quantum de Finetti theorem, where we restrict attention to a special class of symmetric Werner states (see \cite{werner}) as was, in the de Finetti scenario, done e.g. in \cite{ckmr}.
\end{section}
\begin{section}{Notation}
The symbols $\lambda,\mu,\nu,\gamma$ will be used to denote {\bf Young frames}. The row lengths of a young frame $\lambda$ with $d$ rows and $n$ entries will be denoted $\lambda_i$, $i=1,\ldots,d$.
The set of Young frames with at most $d\in\nn$ rows and $n\in\nn$ boxes is denoted $YF_{d,n}$.\\
For a given natural number $N$, the set $\{1,\ldots,N\}$ will be abbreviated $[N]$.\\
{\bf Hilbert spaces} are all assumed to have finite dimension and are over the field $\mathbb C$. The linear space of operators over a Hilbert space $\kr$ is written $\mathcal B(\kr)$.\\
For a natural number $n$, the symbol $\mathbb B$ denotes (with a slight abuse of notation) the standard {\bf representation} of the {\bf symmetric group} $S_n$ on $\hr^{\otimes n}$. The symbol $\hr$ is reserved for a Hilbert space of dimension $d\in\nn$. The unique irreducible representation of $S_n$ corresponding to a Young tableau $\lambda$ will be written $F_\lambda$.\\
The {\bf multiplicity of an irreducible subspace} of $\mathbb B$ corresponding to a Young frame $\lambda$ is given by the dimension of the corresponding irreducible representation $U^d_\lambda$ of the standard representation of the unitary group $U(d)$ on $\hr^{\otimes n}$, it is written $\dim U^d_\lambda$.\\
{\bf Projections onto isotypical subspaces} are written $P_\lambda$ ($\lambda\in YF_{d,n}$). The corresponding 'flat' states are denoted $\pi_\lambda$ and defined through $\pi_\lambda:=\frac{1}{\tr\{P_\lambda\}}P_\lambda$, with $\tr$ denoting the usual trace function on $\hr^{\otimes n}$.\\
The set of {\bf probability distributions} on a finite set $\bX$ is denoted $\mathfrak P(\bX)$, the cardinality
of $\bX$ by $|\bX|$.\\
For two probability distributions $r,s\in\mathfrak P(\bX)$, the distance between them is measured by
$\|r-s\|:=\sum_{x\in\bX}^d|r(x)-s(x)|$.\\
An important entropic quantity is the {\bf relative entropy}. We define it (using {\bf base two logarithm} which is, throughout, written as $\log$).
as follows: Given a finite set $\bX$ and two probability distributions $r,s\in\mathfrak P(\bX)$, the relative entropy $D(r||s)$ is given by
\begin{align}
D(r||s):=\left\{\begin{array}{ll}\sum_{x\in \bX}r(x)\log(r(x)/s(x)),&\mathrm{if}\ s\gg r\\ \infty,&\mathrm{else}\end{array}.\right.
\end{align}
In case that $D(r||s)=\infty$, for a positive number $a>0$, we use the convention $2^{-aD(r||s)}=0$. The relative entropy is connected to $\|\cdot\|$ by the Pinsker's inequality $D(r||s)\leq\frac{1}{2\ln(2)}\|r-s\|^2$.\\
The {\bf binary entropy} of $r\in\mathfrak P(\{0,1\})$ is defined by the formula
\begin{align}
h(r):=-\sum_{x\in \bX}r(x)\log(r(x)).
\end{align}
If $B\subset[n]$ for some $n\in\nn$, then $\tr_B$ denotes the usual {\bf partial trace} functional, $\tr_B:=\otimes_{i=1}^n\E_{\eins_B(i)}$, where $\eins_B$ is the indicator function on $[n]$ and $\E_0:=Id_{\hr}$, $\E_1:=\tr_\hr$. For $A\in\mathcal B(\hr^{\otimes n})$, $\tr_B\{A\}\otimes\pi_B:=\otimes_{i=1}^n\cn_{\eins_B(i)}$, where $\cn_0=Id$ and $\cn_1=T$, with $T$ defined by $T(a):=\pi_\hr\tr\{a\}$ ($a\in\mathcal B(\hr)$). The symbol $\pi_\hr$ denotes the {\bf maximally mixed state} on $\hr$: $\pi_\hr=\frac{1}{d}\eins_\hr$.
\end{section}
\begin{section}{Results and their proofs}
\begin{theorem}\label{mainresult}
Let $\lambda\in YF_{d,n}$. Let $B\subset[n]$. Then for all $\lambda'\in YF_{d,n}$ for which $|\lambda_m'-\lambda_m|>(d-1)|B|$ occurs for an $m\in[d]$ we have
\begin{align}
 \tr\{P_{\lambda'}(\tr_B\{P_\lambda\}\otimes\eins_B)\}=0,
\end{align}
and this is equivalent to the statement that, with the definition
\begin{align}
\overline{\mathbf S}_n:\mathcal B(\hr^{\otimes n})\to\mathcal B(\hr^{\otimes n}),\qquad A\mapsto\sum_{\tau\in S_n}\frac{1}{n!}\mathbb B(\tau)A\mathbb B(\tau^{-1}),
\end{align}
it holds
\begin{align}
 \tr\{P_{\lambda'}\overline{\mathbf S}_n(\tr_B\{P_\lambda\}\otimes\eins_B)\}=0.
\end{align}
\end{theorem}
The following theorem is an easy application of Theorem \ref{mainresult}. We focus on the (more elementary) case of the depolarising \emph{qubit} channel, although our results could be applied to the general case as well.
\begin{theorem}\label{application}
Let $\lambda,\lambda'\in YF_{2,n}$. Let $p\in[0,1]$ and $\cn_p:\mathcal B(\mathbb C^2)\to\mathcal B(\mathbb C^2)$ be the depolarising channel with depolarising probability $p$. Then if $|\lambda_1-\lambda_1'|>n\cdot p$ we have
\begin{align}
\tr\{P_{\lambda'}\cn^{\otimes l}_p(P_\lambda)\}\leq2^{-n(\frac{2}{\log 2}(\frac{1}{n}|\lambda_1-\lambda_1'|-p)^2-\Delta(n))}
\end{align}
for some function $\Delta:\nn\rightarrow\mathbb R_+$ satisfying $\lim_{n\to\infty}\Delta(n)=0$.
\end{theorem}
\begin{remark}
If $|\lambda_1-\lambda_1'|<n\cdot p$, then the approach pursued in this paper does not give good estimates, since the method of proof for Theorem \ref{application} is to apply Theorem \ref{mainresult}, and this only says when $\tr\{P_{\lambda'}\overline{\mathbf S}_n(\tr_B\{P_\lambda\}\otimes\eins_B)\}=0$ holds, without computing estimates on the l.h.s.
\end{remark}
\begin{proof}[Proof of Theorem \ref{application}]
The depolarising qubit channel $\cn_p=p\cdot Id+(1-p)T$ can be written as
\begin{align}
\cn_p=\bar p(0)\cn_0+\bar p(1)\cn_1,
\end{align}
where $\bar p(0)=p$ and $\bar p(1)=1-p$ hold. Then, clearly, with $\bar p^{\otimes n}\in\mathfrak P(\{0,1\}^n)$ defined by $\bar p^{\otimes n}(x^n):=\prod_{i=1}^n\bar p(x_i)$ and $\cn_{x^n}:=\otimes_{i=1}^n\cn_{x_i}$ we get
\begin{align}
\cn_p^{\otimes n}&=\sum_{x^n\in\{0,1\}^n}\bar p^{\otimes n}(x^n)\cn_{x^n}\\
&=\sum_{k=0}^n2^{-(k\log p+(n-k)\log(1-p))}\sum_{x^n\in T_k}\cn_{x^n},
\end{align}
where we additionally used the typical sets $T_k:=\{x^n\in\{0,1\}^n:N(1|x^n)=k\}$, where $N(1|x^n)$ is the number of ones in the word $x^n$. All that is left to do is examine the term $\sum_{x^n\in T_k}\cn_{x^n}$:
\begin{align}
\sum_{x^n\in T_k}\cn_{x^n}(\cdot)&=\sum_{B\subset[n]:|B|=k}\tr_B\{\cdot\}\otimes\pi_B\\
&={n\choose k}\frac{1}{n!}\sum_{\tau\in S_n}\mathbb B(\tau)[\tr_{[k]}\{\mathbb B(\tau^{-1})\cdot\mathbb B(\tau)\}\otimes\pi_{[k]}]\mathbb B(\tau^{-1}).
\end{align}
But $\mathbb B(\tau^{-1})P_\lambda\mathbb B(\tau)=P_\lambda$ for every $\tau\in S_n$, so
\begin{align}
\cn_p^{\otimes n}(P_\lambda)&=\sum_{k=0}^n2^{-(k\log p+(n-k)\log(1-p))}{n\choose k}\frac{1}{n!}\sum_{\tau\in S_n}\mathbb B(\tau)[\tr_{[k]}\{P_\lambda\}\otimes\pi_{[k]}]\mathbb B(\tau^{-1})\\
&=\sum_{k=0}^n2^{-(k\log p+(n-k)\log(1-p))}{n\choose k}\overline{\mathbf{S}}_n[\tr_{[k]}\{P_\lambda\}\otimes\pi_{[k]}].
\end{align}
Especially, this last equality holds as well if we replace $P_\lambda$ by a projection onto \emph{any} subspace that is irreducible under the action of $\mathbb B$. It now follows, by the estimate ${n\choose k}\leq2^{-n(h(\bar k)-\Delta_1(n))}$ that can be found e.g. in \cite{csiszar-koerner} (here, $\Delta_1:\nn\to\rr_+$ satisfies $\lim_{n\to\infty}\Delta_1(n)=0$ and $\bar k$ is (for every $k\in\nn$) the distribution $\bar k\in\mathfrak P(\{0,1\})$ introduced as $\bar k(0):=k/n$) and using Pinsker's inequality
\begin{align}
 \tr\{P_{\lambda'}\cn_p^{\otimes n}(\pi_\lambda)\}&\leq\sum_{k=0}^n2^{-n(D(\bar k||\bar p)-\Delta_1(n))}\tr\{P_{\lambda'}\overline{\mathbf S_n}(\tr_{[k]}\{\pi_\lambda\}\otimes\pi_{n-k})\}\\
&\leq n\sum_{k:|\lambda_1-\lambda'_1|\leq k}^n2^{-n(D(\bar k||p)-\Delta_1(n))}\\
&\leq2^{-n(\frac{2}{\ln2}(\frac{1}{n}|\lambda_1-\lambda_1'|-p)^2-\Delta(n))}.
\end{align}
The function $\Delta$ is defined by $\Delta(n):=\Delta_1(n)+\frac{1}{n}\log(n)$.
\end{proof}
\begin{proof}[Proof of Theorem \ref{mainresult}]
First, let $B=\{l+1,\ldots,n\}$ for some $l\in[n]$ and set $k:=n-l$, $A=[l]$. We will see later, that this is without loss of generality. Then we have a natural action of $S_l\times S_k$
on $\hr_A\otimes \hr_B:=\hr^{\otimes k}\otimes\hr^{\otimes l}$ as
\begin{align}\label{eqn-1}
 \mathbb B^{A\times B}(\pi,\tau)v=\mathbb B^A(\pi)\otimes\mathbb B^B(\tau)v\qquad (\pi\in S_l,\ \tau\in S_k,\ v\in \hr_A\otimes\hr_B).
\end{align}
It is clear that $S_l\times S_k\subset S_n$ holds and, by equation (\ref{eqn-1}) it is also clear that
\begin{align}
\mathbb B\restriction _{S_l\times S_k}= \mathbb B^{A\times B},
\end{align}
so the irreducible subspaces of $\mathbb B$ can be decomposed into direct sums of irreducible subspaces of $\mathbb B^{A\times B}$. Prototypes of the latter can easily
be constructed. Take an irreducible subspace $V_\mu$ of $\mathbb B^A$ and another one, $V_\nu$, of $\mathbb B^k$. Then $V_\mu\otimes V_\nu$ is an irreducible subspace of
$\mathbb B^{A\times B}$.\\
However, not every irreducible subspace of $\mathbb B^{A\times B}$ arises in that way, as can easily be seen by considering $n=2$, $l=k=1$ and taking the irreducible subspace
$V_{+}=\mathrm{span}(\{e_1\otimes e_2+e_2\otimes e_1\})$. It is well-known, that this subspace can not be decomposed as $V_+=V_1\otimes V_2$. What \emph{can} be said in this context
is the following.\\
If the Littlewood-Richardson coefficient $c^\lambda_{\mu\nu}\neq0$ for some choice of $\lambda\in YF_{d,n},\ \mu\in YF_{d,l},\ \nu\in YF_{d,k}$ then each \emph{irreducible} subspace
$\mathcal V_{\lambda}$ of $\mathbb B$ contains one irreducible subspace $\mathcal V_{\mu\nu}$ of $\mathbb B^{A\times B}$ with dimension $\dim(\mathcal V_{\mu\nu})=\dim(F_\mu)\cdot\dim(F_\nu)$.
The latter is, itself, part of the \emph{isotypical} subspace $V_{\mu\nu}$ of $\mathbb B^{A\times B}$ and this subspace indeed satisfies
\begin{align}
V_{\mu\nu}=V_\mu\otimes V_\nu.
\end{align}
So, what can be said is the following. For each $\lambda\in YF_{d,n}$ we have
\begin{align}
 V_\lambda\subset \bigoplus_{\mu\in YF_{d,l},\ \nu\in YF_{d,k}}\mathbbm 1_{\{\mu,\nu:c^\lambda_{\mu\nu}\neq0\}}(\mu,\nu)V_\mu\otimes V_\nu.
\end{align}
This also implies that
\begin{align}
 P_\lambda\leq\sum_{\mu\in YF_{d,l},\ \nu\in YF_{d,k}}\mathbbm 1_{\{\mu,\nu:c^\lambda_{\mu\nu}\neq0\}}(\mu,\nu)P_\mu\otimes P_\nu.\label{eqn-2}
\end{align}
Using this approach, it is possible to obtain the necessary estimates to prove Theorem \ref{mainresult}. A more elegant way is to use results of \cite{ckmr} (we thank M. Christandl for making them known to us), stating that
\begin{align}
\tr_{[k]}\{P_\lambda\}=\dim U^d_\lambda\sum_{\nu,\mu}c^\lambda_{\mu\nu}\frac{\dim F_\nu}{\dim U^d_\mu}P_\mu.
\end{align}
Keeping in mind that $\supp(\tr_{[k]}\{P_\lambda\}\otimes\pi_{[k]})\subset\supp(\overline{\mathbf S}_n[\tr_{[k]}\{P_\lambda\}\otimes\pi_{[k]}])$ holds, it is clear that Theorem \ref{mainresult} can possibly be proven by using estimates on the r.h.s. rather than the l.h.s., and indeed this is true and we will follow this idea. Now, by linearity of $\overline{\mathbf S}_n$
\begin{align}
\overline{\mathbf S}_n[\tr_{[k]}\{P_\lambda\}\otimes\eins_{[k]}]=\dim U^d_\lambda\sum_{\nu,\mu,\gamma}c^\lambda_{\mu\nu}\frac{\dim F_\nu}{\dim U^d_\mu}\overline{\mathbf S}_n[P_\mu\otimes P_\gamma]
\end{align}
and, by invariance of $\overline{\mathbf S}_n[P_\mu\otimes P_\gamma]$ under the usual product action $U^{\otimes n}$ of $U^d$ on $\hr^{\otimes n}$ we get
\begin{align}
\overline{\mathbf S}_n[P_\mu\otimes P_\gamma]&=\sum_{\lambda'} \alpha^{\lambda'}_{\mu\gamma}P_{\lambda'}
\end{align}
for some set of coefficients $\alpha^{\lambda'}_{\mu\nu}$. Obviously, $\alpha^{\lambda'}_{\mu\gamma}=0\Leftrightarrow c^{\lambda'}_{\mu\gamma}=0$. We know even more:
\begin{align}
\alpha^{\lambda'}_{\mu\gamma}&=\frac{1}{\dim F_{\lambda'}}\tr\{P_{\lambda'}\overline{\mathbf S}_n[P_\mu\otimes P_\gamma]\}\\
&=\frac{1}{\dim F_{\lambda'}}\tr\{P_{\lambda'}(P_\mu\otimes P_\gamma)\}\\
&=c^{\lambda'}_{\mu\gamma}\frac{\dim F_\mu\dim F_\gamma}{\dim F_{\lambda'}}\dim U^d_{\lambda'},
\end{align}
so
\begin{align}
\overline{\mathbf S}_n[\tr_{[k]}\{P_\lambda\}\otimes\eins_{[k]}]=\sum_{\nu,\mu,\gamma,\lambda'}c^\lambda_{\mu\nu}c^{\lambda'}_{\mu\gamma}\cdot(\frac{\dim U^d_\lambda\dim U^d_{\lambda'}}{\dim U^d_\mu}\cdot\frac{\dim F_\nu\dim F_\mu\dim F_\gamma}{\dim F_{\lambda'}} )\cdot P_{\lambda'}.
\end{align}
At this point, it should in principle be possible to get better results by solving, for arbitrary $\lambda,\lambda'\in YF_{d,n}$, the optimization problems
\begin{align}
X_{\lambda,\lambda'}&:=\max\{\dim F_\nu\dim F_\mu\dim F_\gamma:c^\lambda_{\mu\nu}c^{\lambda'}_{\mu\gamma}\neq0\ \wedge\ \mu\in YF_{d,l}\ \wedge\ \nu,\gamma\in YF_{d,k}\},\\
Y_{\lambda,\lambda'}&:=\min\{\dim F_\nu\dim F_\mu\dim F_\gamma:c^\lambda_{\mu\nu}c^{\lambda'}_{\mu\gamma}\neq0\ \wedge\ \mu\in YF_{d,l}\ \wedge\ \nu,\gamma\in YF_{d,k}\}.
\end{align}
Bounds, especially anything better than the trivial upper bound '$X_{\lambda\lambda'}\leq poly(n)\cdot\dim F_\lambda\cdot2^k$', would, at least from our perspective, be of further interest for the question at hand. We come back to the problem after completing this proof.\\
Now letting $\mu,\nu,\gamma$ denote Young frames that contribute to the above sum (meaning that $c^\lambda_{\mu\nu}\neq0$ and $\gamma\in YF_{d,k}$ holds), we ask what a Young frame $\lambda'\in YF_{d,n}$ has to fulfill in order for $c^{\lambda'}_{\mu\gamma}\neq0$.\\
We know from \cite{christandl-cc} that $c^\lambda_{\mu\nu}\neq0$ if and only if there exist nonnegative operators $\mathfrak A,\mathfrak B,\mathfrak C\in\B(\hr)$ such that
\begin{align}
\mathfrak A+\mathfrak B=\mathfrak C,\ \ \spec(\mathfrak A)=\mu,\ \ \spec(\mathfrak B)=\nu,\ \ \spec(\mathfrak C)=\lambda
\end{align}
hold. For the same reason, if $\lambda'$ shall fulfill $c^{\lambda'}_{\mu\gamma}\neq0$, then there have to exist nonnegative operators $\mathfrak D,\mathfrak C'\in\B(\hr)$ such that
\begin{align}
 \mathfrak A+\mathfrak D=\mathfrak C',\ \ \spec(\mathfrak D)=\gamma,\ \ \spec(\mathfrak C')=\lambda'
\end{align}
hold. The relations between all these spectra are governed by Horn's inequalities \cite{horn} (see e.g. \cite{bhatia-story,fulton} to get some feeling for the topic), and from these we only need the very basic ones, namely:
\begin{align}\label{eqn:horns1}
\lambda_m\leq\mu_i+\nu_j&\qquad\forall\ i,j,m\in[d]\ :\ m=i+j-1\\
&\mathrm{and}\nonumber\\
\label{eqn:horns2}\lambda'_m\leq\mu_i+\gamma_j&\qquad\forall\ i,j,m\in[d]\ :\ m=i+j-1.
\end{align}
In combination, this yields for every $m\in[d]$ and $i,j\in[d]$ such that $m=i+j-1$:
\begin{align}
\lambda_m-\lambda_m'&\leq(\mu_i+\nu_j)-n+\sum_{r\in[d]:r\neq m}\lambda'_r\\
&\leq\mu_m+\nu_1-n+\sum_{r\in[d]:r\neq m}(\mu_{r}+\gamma_1)\\
&=\sum_{l=1}^d\mu_l+\nu_1-n+(d-1)\gamma_1)\\
&\leq l-n+d\cdot k\\
&\leq (d-1)\cdot k,
\end{align}
as well as
\begin{align}
\lambda_m'-\lambda_m&\leq(\mu_i+\gamma_j)-n+\sum_{r\in[d]:r\neq m}\lambda_r\\
&\leq\mu_m+\gamma_1-n+\nu_1+\sum_{r\in[d]:r\neq m}\mu_{r}\\
&\leq l-n+d\cdot k\\
&\leq (d-1)\cdot k,
\end{align}
so $|\lambda_m-\lambda'_m|\leq (d-1)|B|$ for all $m\in[d]$, as was to be proven. The whole argument is completely independent of where we make the cut between $A$ and $B$, hence Theorem \ref{mainresult}.
\end{proof}
\begin{lemma}
If, for $\lambda\in YF_{2,n}$ it holds $\lambda_1=n$ then for all $\lambda'\in YF_{2,n}$ we have
\begin{align}
X_{\lambda\lambda'}\leq2^{k\cdot h(\lambda'_2/k)},
\end{align}
if $\lambda'_2/k\in[0,1]$ and $X_{\lambda\lambda'}=0$, else.
\end{lemma}
\begin{remark}
This also shows that (not too surprisingly) the output of the depolarising channel, given the input $\pi_\lambda$ with $\lambda_1=n$, is concentrated around symmetric Werner states satisfying $\lambda_1'\approx n\cdot\frac{p}{2}$. More interesting properties remain to be investigated.
\end{remark}
\begin{proof}
By the arguments from the previous proof, especially inequalities (\ref{eqn:horns1}) and (\ref{eqn:horns2}) we know that $c^{\lambda}_{\mu\nu}\neq0$ can only hold if $\mu_1=l$ and $\nu_1=k$. Hence $\dim F_\mu=\dim F_\nu=1$. For $\gamma$ to satisfy $c^{\lambda'}_{\mu\gamma}\neq0$ this implies
\begin{align}
\lambda'_1\leq l+\gamma_1,\qquad\lambda'_2\leq l+\gamma_2,\qquad\lambda'_2\leq\gamma_1.
\end{align}
Obviously, $\dim F_\gamma$ strictly decreases while $\gamma_1$ increases, hence setting $\gamma_1=\lambda_2$ yields an upper bound on $\dim F_\gamma$. Then, invoking inequality (30) from \cite{noetzel} we get
\begin{align}
\dim F_\gamma\leq2^{k\cdot h(\gamma_1/k)}=2^{k\cdot h(\lambda'_2/k)}.
\end{align}
\end{proof}
\end{section}
\begin{acknowledgement}
Many thanks go to Matthias Christandl for stimulating discussions about representation theory and pointing out to us a more elegant way of proof. We also want to thank Andreas Winter for sharing with us his knowledge about some more recent results concerning the depolarising channel. This work was supported by the BMBF via grant 01BQ1050.
\end{acknowledgement}


\begin{thebibliography}{99}
\bibitem{abbn} R. Ahlswede, I. Bjelakovic, H. Boche, J. N\"otzel ``Quantum capacity under adversarial noise: arbitrarily varying quantum channels'', \emph{Comm. Math. Phys.}, Vol. 317, Iss. 1, 103-156 (2013)

\bibitem{bhatia-story} R. Bhatia, ``Linear Algebra to Quantum Cohomology: The Story of Alfred Horn's Inequalities'', \emph{Amer. Math. Monthly}, Vol. 108, No. 4., 289-318 (2001)

\bibitem{calderbank-rains-shor-sloane} A. R. Calderbank, E. M. Rains, P. W. Shor, N. J. A. Sloane, ``Quantum Error Correction and Orthogonal Geometry'', \emph{Phys. Rev. Lett.} Vol. 78, 405-408 (1997)

\bibitem{ckmr} M. Christandl, R. K\"onig, G. Mitchison, R. Renner, ``One-and-a-half quantum de Finetti theorems'', \emph{Comm. Math. Phys.} Vol. 273, Iss. 2, 473-498 (2007)

\bibitem{christandl-cc}
M. Christandl, ``A quantum information-theoretic proof of the relation between Horn’s problem and the Littlewood-Richardson coefficients'',
 \emph{LNCS}, Vol. 5028, 120-128, (2008)

\bibitem{csiszar-koerner}
I. Csiszar, J. K\"orner, \emph{Information Theory; Coding Theorems for Discrete Memoryless Systems}, Akad\'emiai Kiad\'o, Budapest/Academic Press Inc., New York 1981

\bibitem{fern-whaley} J. Fern, K.B. Whaley, ``Lower bounds on the nonzero capacity of Pauli channels'', \emph{Phys. Rev. A} Vol. 78, 062335-062344 (2008)

\bibitem{fulton} W. Fulton, ``Eigenvalues, invariant factors, highest weights, and Schubert calculus'', \emph{Bull. Amer. Math. Soc. (N.S.)} Vol. 37, No. 3, 209-249 (2000)

\bibitem{horn} A. Horn, ``Eigenvalues of sums of Hermitian matrices'', \emph{Pacific J. Math.} Vol. 12, 225–241 (1962)

\bibitem{klyachko} A.A. Klyachko: ``Stable bundles, representation theory and Hermitian operators'', \emph{Sel. math. New. ser.} 4, 419–445 (1998)

\bibitem{lidskii} B.V. Lidskii: ``Spectral polyhedron of the sum of two Hermitian matrices'', \emph{Func. Anal. Appl.} 16, 139–140 (1982)

\bibitem{ouyang} Y. Ouyang, ``Upper bounds on the quantum capacity of some quantum channels using the coherent information of other channels'', \emph{ 	arXiv:1106.2337} (2011)

\bibitem{werner} R.F. Werner, ``Quantum states with Einstein-Podolsky-Rosen correlations admitting a hidden-variable model'', \emph{Phys. Rev. A} Vol. 40, 4277–4281 (1989)

\bibitem{noetzel} J. N\"otzel, ``A solution to two party typicality using representation theory of the symmetric group'', \emph{arXiv:1209.5094} (2012)
\end{thebibliography}
\end{document}